\documentclass{article}
\usepackage[utf8]{inputenc}

\usepackage{xspace}
\usepackage{tkz-euclide}
\usepackage{tikz}
\usetikzlibrary{backgrounds,patterns,3d,calc}
\usepackage{xinttools}

\def\STANDALONE{true}
\usepackage{cancel}

\usepackage{algorithmic}
\usepackage{algorithm}
\usepackage{fullpage}
\ifdefined\STANDALONE{}
\usepackage[subpreambles]{standalone}
\fi
\usepackage{geometry}
\usepackage{subcaption}
\usepackage{amsmath}
\usepackage{amsthm}
\usepackage{amssymb}
\usepackage{esint}
\usepackage[unicode=true,
 bookmarks=false,
 breaklinks=false,pdfborder={0 0 1},backref=section,colorlinks=false]
 {hyperref}

\theoremstyle{plain}
\newtheorem{thm}{Theorem}
\newtheorem{lem}[thm]{Lemma}
\newtheorem{lemma}[thm]{Lemma}
\newtheorem*{fact*}{\Fact}

\theoremstyle{definition}
\newtheorem{definition}[thm]{Definition}

\newtheorem{fact}[thm]{Fact}
\theoremstyle{remark}
\newtheorem{remark}[thm]{Remark}

\newcommand{\br}[1]{\left(#1\right)}

\newcommand{\brq}[1]{\left[ #1 \right]}

\newcommand{\eps}{\varepsilon}

\newcommand{\size}[1]{\left|#1\right|}
\newcommand{\setst}[2]{\left\{  #1\left|#2\right.\right\}  }
\newcommand{\set}[1]{\left\{  #1\right\}  }

\newcommand{\cost}[2]{cost\br{#1,#2}}

\newcommand{\ceil}[1]{\left\lceil{} #1\right\rceil{} }
\newcommand{\algunc}{\ensuremath{ALG_{unc}^{\ell}\br{d}}}

\newcommand{\OO}{\mathcal{O}}

\newcommand{\NP}{\textsf{NP}\xspace}

\newcommand{\FPT}{\textsf{FPT}\xspace}
\newcommand{\W}{\textsf{W}[1]\xspace}
\newcommand{\WW}{\textsf{W}[2]\xspace}

\newcommand{\UKM}{\textsc{Uncapacitated KM}\xspace}
\newcommand{\CKM}{\textsc{CKM}\xspace}
\newcommand{\UCKM}{\textsc{Uniform CKM}\xspace}
\newcommand{\NUCKM}{\textsc{Non-Uniform CKM}\xspace}

\usepackage{color}

\ifdefined\DEBUG{}
\definecolor{marekgreen}{RGB}{0,185,0}
\definecolor{orange}{RGB}{255,128,0}
\definecolor{orange}{RGB}{255,128,0}
\definecolor{orange}{RGB}{255,128,0}
\definecolor{orange}{RGB}{255,128,0}

\newcommand{\mar}[1]{\textcolor{marekgreen}{#1}}
\newcommand{\mic}[1]{{\color{blue}{#1}}}

\newcommand{\mes}[1]{\textcolor{cyan}{#1}}

\def\rem#1{{\marginpar{\raggedright\scriptsize #1}}}
\newcommand{\jarr}[1]{\rem{\textcolor{red}{\(\bullet \) #1}}}
\newcommand{\marr}[1]{\rem{\textcolor{marekgreen}{\(\bullet \) #1}}}
\newcommand{\micr}[1]{\rem{\textcolor{blue}{\(\bullet \) #1}}}
\newcommand{\janr}[1]{\rem{\textcolor{orange}{\(\bullet \) #1}}}
\newcommand{\mesr}[1]{\rem{\textcolor{cyan}{\(\bullet \) #1}}}

\else

\newcommand{\mar}[1]{#1}
\newcommand{\mic}[1]{#1}

\newcommand{\mes}[1]{#1}
\newcommand{\jarr}[1]{ }
\newcommand{\marr}[1]{ }
\newcommand{\micr}[1]{ }
\newcommand{\janr}[1]{ }
\newcommand{\mesr}[1]{ }

\fi

\title{\mic{Constant factor FPT approximation for capacitated k-median
}}
\date{}

\author{Marek	Adamczyk	\footnote{marek.adamczyk@mimuw.edu.pl} \\ University of Warsaw
\and Jarosław	Byrka \footnote{jby@cs.uni.wroc.pl}\\ University of Wrocław
\and Jan	Marcinkowski \footnote{jasiekmarc@cs.uni.wroc.pl} \\University of Wrocław
\and Syed M. Meesum \footnote{syedmohammad.meesum@uwr.edu.pl}\\ University of Wrocław
\and Michał	Włodarczyk \footnote{m.wlodarczyk@mimuw.edu.pl}\\ University of Warsaw
}

\makeatother

\begin{document}
  \maketitle{}

  \begin{abstract}
    Capacitated k-median is one of the few outstanding optimization problems for
    which the existence of a polynomial time constant factor approximation
    algorithm remains an open problem. In a series of recent papers algorithms
    producing solutions violating either the number of facilities or the
    capacity by a multiplicative factor were obtained. However, to produce
    solutions without violations appears to be hard and potentially requires
    different algorithmic techniques. Notably, if \mes{parameterized} by the number of
    facilities \(k\), the problem is also \(W[2]\) hard, making the existence of an
    exact FPT algorithm unlikely. In this work we provide an FPT-time constant
    factor approximation algorithm preserving both cardinality and capacity of
    the facilities.
\mic{The algorithm runs in time $2^{\OO(k\log k)}n^{\OO(1)}$ and achieves an approximation ratio of $7+\eps$.}
  \end{abstract}


\section{Introduction}\label{sec:intro}


For many years approximation algorithms and FPT algorithms were developed in parallel. Recently the two paradigms are being combined and provide intriguing discoveries in the intersection of the two worlds. It is particularly interesting in the case of problems for which we fail to progress improving the approximation ratios in polynomial time. An excellent example of such a combination is the FPT approximation algorithm for the \textsc{$k$-Cut} problem by Gupta et al.~\cite{gupta2018fpt}.

In this work we focus on the \textsc{Capacitated $k$-Median} problem, whose approximability attracted attention of many researchers. Unlike in the case of the \textsc{$k$-Cut} problem, it is still not clear what approximation is possible for \textsc{Capacitated $k$-Median} in polynomial time.
As shall be discussed in more detail in the following section, the best true approximation known is $\OO(\log k )$ based on tree embedding of the underlying metric. The other algorithms either violate the bound on the number of facilities or the capacity constraints.

Our main result is a $(7 + \epsilon)$-approximation algorithm for the \textsc{Capacitated $k$-Median} problem running in FPT($k$) time, that exploits techniques from both ---  approximation and \FPT{} --- realms.
The algorithm builds on the idea of clustering  the clients into \mar{$\ell =  \OO(k\cdot (\log n) / \eps)$} locations, \mar{which is similar to the approach from the $\OO(\log k )$-approximation algorithm, where one creates $\OO\br{k}$ clusters}. This is followed by guessing the distribution of the $k$ facilities inside these $\ell$ clusters.
Having such a structure revealed, we simplify the instance further by rounding particular distances and reduce the problem to linear programming  over a totally unimodular matrix.

\subsection{Problems overview and previous work}


    In the \textsc{Capacitated \(k\)-Median} problem (\CKM{}), we are given a
    set \(F\) of facilities, each facility \(f\) with a capacity \(u_f \in
    \mathbb{Z}_{\geqslant 0}\), a set \(C\) of clients, a  metric \(d\) over \(F
    \cup C\) and an upper bound \(k\) on the number of facilities we can open.
    A solution to the \CKM problem is a set \(S \subseteq F\) of at most \(k\) open facilities
    and a connection assignment \(\phi: C \to S\) of clients to open facilities
    such that \(\size{\phi^{-1}(f)} \leqslant u_f\) for every facility \(f \in S\).
    The goal of the problem is to find a solution that minimizes the connection cost \(\sum_{c \in C}d(c, \phi(c))\).

    In the case when all the facilities can
    serve at most \(u\) clients, for some integer \(u\), we obtain the
    \textsc{Uniform \CKM} problem. \marr{Although our \((7+\eps)\)-approximation for the non-uniform case yields a result for this version as well, we present a derivation for this case separately as it allows for a gradual presentation of our ideas.}

    \paragraph{Uncapacitated \(k\)-median} The standard \(k\)-median problem,
    where there is no restriction on the number of clients served by a facility,
    can be approximated up to a constant factor~\cite{charikar1999constant,
    arya2004local}. The current best is the \((2.675+\epsilon)\)-approximation
    algorithm of Byrka et al.~\cite{byrka2015improved}, which is a result of
    optimizing a part of the algorithm by Li and
    Svensson~\cite{li2013approximating}.

    \paragraph{Approximability of \CKM{}} As already stressed, \textsc{Capacitated \(k\)-Median} is among
    few remaining fundamental optimization problems for which it is not clear if
    there exist polynomial time constant factor approximation algorithms. All
    the known algorithms violate either the number of facilities or the
    capacities. In particular, already the algorithm of Charikar et
    al.~\cite{charikar1999constant} gave 16-approximate solution for the uniform
    capacitated \(k\)-median violating the capacities by a factor of 3. Then
    Chuzhoy and Rabani~\cite{chuzhoy2005approximating} considered general
    capacities and gave a 50-approximation algorithm violating capacities by a
    factor of 40.

    The difficulty appears to be related to the unbounded integrality gap of the
    standard LP relaxation. To obtain integral solutions that are bounded with respect to
    the fractional solution to the standard LP, one has to either allow the
    integral solution to open twice as much facilities or to violate the capacities
    by a factor of two. LP-rounding algorithms essentially matching these limits
    have been obtained~\cite{aardal2015approximation, BFRS15}.

    Subsequently,  Li  broke  this  integrality  gap  barrier  by  giving  a
    constant  factor  algorithm  for the capacitated k-median by opening \((1 +
    \eps) \cdot k\) facilities~\cite{Li15uniform,Li16non_uniform}. Afterwards
    analogous results, but violating the capacities by a factor of \((1 +
    \eps)\) were also
    obtained~\cite{DBLP:conf/ipco/ByrkaRU16,DBLP:conf/icalp/DemirciL16}.

    The algorithms with \((1+\eps)\) violations are all based on strong LP
    relaxations containing additional constraints for subsets of facilities.
    Notably, it is not clear if these relaxations can be solved exactly in
    polynomial time, still they suffice to construct an approximation algorithm
    via the ``round-or-separate” technique that iteratively adds consistency
    constraints for selected subsets. Although while spectacularly breaking the
    standard LP integrality bound, these techniques appear insufficient to yield
    a proper approximation algorithm that does not violate constraints.

\mar{
The only true approximation for \CKM{} known is a folklore \(\OO\br{\log k}\) approximation algorithm
    that can be obtained via the metric tree embedding with expected logarithmic
    distortion~\cite{FakcharoenpholRT03}.}
\mic{To the best of our knowledge,
this result has not been explicitly published, but it can be obtained similarly to the \(\OO\br{\log k}\)-approximation for \UKM by Charikar~\cite{DBLP:conf/stoc/CharikarCGG98}.
For the sake of completeness and since it follows easily from our framework, we give its proof in Section~\ref{sec:folklore} without claiming credit for it. } \marr{Previously it was written that the $\lg k$-approximation was indicated in the work of~\cite{charikar1999constant} but after looking at this paper I don't think it's the case}

    This \(\OO\br{\log k}\)-approximation is in contrast with other
    capacitated clustering problems such as facility location and k-center,
    for which constant factor approximation algorithms are
    known~\cite{korupolu2000analysis,cygan2012lp}.


\subsection{Parameterized Complexity}
A parameterized problem instance is
    created by associating an input instance with an integer parameter \(k\). We
    say that a problem is \emph{fixed parameter tractable} (\FPT{}) if any
    instance \((I, k)\) of the problem can be solved in time
    \(f(k)\cdot |I|^{\OO(1)}\), where \(f\) is an arbitrary computable function of
    \(k\).
    We say that a problem is \FPT{} if
    it is possible to give an algorithm that solves it in running time of the
    required form. Such an algorithm we shall call a \emph{parameterized algorithm}.

    To show that a
    problem is unlikely to be \FPT{}, we use parameterized reductions analogous
    to those employed in the classic complexity theory. Here, the concept of
    \textsf{W}-hardness replaces the one of \NP-hardness, and we need not only to
    construct an equivalent instance in \FPT{} time, but also ensure that the
    size of the parameter in the new instance depends only on the size of the
    parameter in the original instance.
In contrast to the \NP-hardness theory, there is a hierarchy of classes $\FPT = \textsf{W[0]} \subseteq \textsf{W[1]} \subseteq \textsf{W[2]} \subseteq \dots$ and these containments are believed to be strict.
If there exists a parameterized reduction transforming a problem known to
    be \textsf{W[t]}-hard for $t>0$ to another problem \(\Pi \), then the problem \(\Pi \) is
    \textsf{W[t]}-hard as well. This provides an argument that \(\Pi \) is unlikely to admit an algorithm with
    running time \(f(k)\cdot |I|^{\OO(1)}\).

\paragraph{} We begin with an argument that allowing \FPT time  for (even uncapacitated) \textsc{$k$-Median} should not help in finding the optimal solution and we still need to settle for approximation.

\begin{fact}
The \textsc{Uncapacitated $k$-Median} problem is \WW-hard when parameterized by $k$, even on metrics induced by unweighted graphs.
\end{fact}
\begin{proof}
Consider an instance of the \textsc{Dominating Set} problem,
which is \WW-hard when parameterized by the solution size.
A dominating set of size at most $k$ exists in graph $G$ if and only if we can find a vertex set $S$ of size $k$, such that all other vertices are at distance 1 from $S$.
This is equivalent to the solution to \textsc{Uncapacitated $k$-Median} on the metric induced by $G$ being of size exactly $|V(G)|-k$.
\end{proof}

\paragraph{Parameterized Approximation}
In recent years new research directions emerged in the intersection of the theory of approximation algorithms and the \FPT{} theory.
It turned out that for some problems that are intractable in the exact sense, parameterization still comes in useful when we want to reduce the approximation ratio.
Some examples are $(2-\eps)$-approximation for \textsc{$k$-Cut}~\cite{gupta2018fpt} or
$f(\mathcal{F})$-approximation for \textsc{Planar-$\mathcal{F}$ Deletion}~\cite{FLMS12} for some implicit function $f$.
The dependency on $\mathcal{F}$ was later improved, leading to
$\OO(\log k)$-approximations for, e.g., \textsc{$k$-Vertex Separator}\cite{Lee18} and \textsc{$k$-Treewidth Deletion}~\cite{gupta-arxiv}.

On the other hand some problems parameterized by the solution size have been proven resistant to such improvements.
Chalermsook et al.~\cite{pasin-focs17} observed that under the assumption of $\textsf{Gap-ETH}$ there can be no parametrized approximation with ratio $o(k)$ for \textsc{$k$-Clique} and none with ratio
$f(k)$ for \textsc{$k$-Dominating Set} (for any function $f$).
Subsequently $\textsf{Gap-ETH}$ has been replaced with a better established hardness assumption $\FPT \ne \W$ for \textsc{$k$-Dominating Set}~\cite{pasin-stoc18}.

\subsection{Organization of the paper}
\mar{
Our main result is stated in Theorem~\ref{thm:nuckm} (Section~\ref{sec:nonuniform}), where we present a $(7 + \epsilon)$-approximation algorithm for the \NUCKM problem running in FPT($k$) time.
}

\mar{To obtain this result we need two ingredients.
First is a metric embedding that reduces the problem to a simpler instance, called $\ell$-centered, what is described in Section~\ref{sec:l-centered}.
This reduction provides a richer structure, which can be exploited to obtain an $\OO\br{\log k}$-approximation via tree embeddings~\cite{FakcharoenpholRT03}. As already mentioned, similar approach was presented by Charikar et al.~\cite{DBLP:conf/stoc/CharikarCGG98} in their algorithm for the uncapacitated setting. We present this result for the sake of completeness in Section~\ref{sec:folklore}, after the main result.
}

\mic{
The second ingredient is a parameterized algorithm for the $\ell$-centered instances. Since it is simpler in the uniform setting, we solve it in Section~\ref{sec:uniform} as a warm up before the main result.
This way the new ideas are being revealed gradually to the reader. 
}

\section{\(\ell \)-Centered instances}\label{sec:l-centered}

	Suppose we work with a graph on nodes \(F\cup C\), on which we are given a
	metric \(d\). In our considerations the set \(F\cup C\) will be fixed
	throughout, however we will be modifying the metric over it.
	Consider an algorithm \(ALG\) which  produces a solution
	\(ALG\br d\) for a metric \(d\).\micr{What about not using $ALG$ and $ALG(d)$ but just $SOL$ for a solution?}
    This solution can be seen as a mapping which we explicitly denote by
	\(\phi^{ALG\br d}\). Its cost in the metric \(d'\) equals
	\(\sum_{c\in C}d'\br{c,\phi^{ALG\br{d}}}\) which we shall briefly denote by
	\(cost\br{\phi^{ALG\br{d}},d'}\).
    \mic{The second argument is useful, when an algorithm \(ALG\)
	produces a solution (mapping) \(ALG\br{d}\) with respect to metric \(d\), but later
	on we may be interested in its cost over a different metric.}
    Also, let
	\(OPT\br{d}\) denote the optimum solution for the \(\CKM{}\) problem on metric
	\(d\).

	In order to solve \CKM{}, we shall invoke an algorithm for
	\UKM{} as a subroutine. Let \(ALG^\ell_{unc}\br{d}\) be a relaxed solution that opens up to $\ell \ge k$ facilities and can break the capacity constraints.
    It induces a mapping
	which, for consistency, we shall denote by \(\phi^{ALG^\ell_{unc}\br{d}}\).
Observe that in this mapping every client can be connected to the closest open facility.
\mic{Since \UKM{} admits constant approximation algorithms, we can work with solutions satisfying:
 \(\cost{\phi^{ALG^\ell_{unc}\br d}}d
 =O\br{\cost{\phi^{OPT\br d}}d} \).
 The larger $\ell$ we allow in the relaxation,
 the smaller constant we will be able to achieve
 in the relation above.}

	Using such an algorithm for \UKM{} as a subroutine, we can find a simpler metric to work with.
    First we build a graph which will induce the metric. Let
	$F(\algunc)$ be the set of facilities opened by $\algunc$. For each
such a facility $f$ we create a copy vertex \(s^f\), which is at
	distance \(0\) from~\(f\). We denote the set of copies by \(S\), i.e., \(S =
	\setst{s^f}{f\in F(\algunc)}\). Given that we demand the distance from
	\(f\) to \(s^f\) to be \(0\), we can naturally extend the metric \(d\) to
	the set \(C\cup F \cup S\). To distinguish facilities from
	$F(\algunc)$ from their copies \(S\), we shall call each copy \(s\in S\) a \emph{center}.

	We build a complete graph on \(S\) and preserve the metric $d$ therein. For every node \(v \not\in S\), be it either a
	client from \(C\) or a facility from \(F\), we place an edge to the closest
	(according to the extended \(d\)) \marr{to jest w sumie oczywiste ale trzeba napisac ze $\algunc$ tez przypisuje do najblizszych centrow klientow; to jest kluczowe zalozenie w pozniejszym lemacie 2}center \(s^{v} \in S\) and set its length to \(d\br{v,s^{v}}\). We call such a graph \(\ell\)-centered and refer to its induced metric as \(d_{\ell}\).

	\begin{definition}\label{def:l-centered}
		An instance of \CKM{} is called \(\ell \)-\emph{centered} if the
		metric, which we shall denote by \(d_{\ell}\), is induced by a weighted
		graph \(G(F\cup C\cup S,E)\) such that
		\begin{enumerate}
			\item \(|S|\leqslant\ell\),
			\item \(\binom{S}{2} \subseteq E\), i.e., \(S\) forms a clique,
			\item for every \(v \in C\cup F\) there is only one edge incident to
				\(v\) in \(E\), and it connects \(v\) to some \(s^v \in S\).
		\end{enumerate}
	\end{definition}

	\begin{figure}
		\centering{}
        \ifdefined\STANDALONE
			\def\biglen{20cm} 
\tikzset{
	half plane/.style={ to path={
			 ($(\tikztostart)!.5!(\tikztotarget)!#1!(\tikztotarget)!\biglen!90:(\tikztotarget)$)
		-- ($(\tikztostart)!.5!(\tikztotarget)!#1!(\tikztotarget)!\biglen!-90:(\tikztotarget)$)
		-- ([turn]0,2*\biglen) -- ([turn]0,2*\biglen) -- cycle}},
	half plane/.default={1pt}
}

\def\maxxy{7} 
\begin{tikzpicture}[x={(-0.5cm,-0.5cm)}, y={(1cm,0cm)}, z={(0cm,1cm)},
				every node/.style={draw,circle,fill=darkgray,inner sep=2pt},
				every label/.append style={rectangle, font=\Large}]
		\tikzstyle{fac}=[draw,circle,fill=lightgray,inner sep=1pt]
		\tikzstyle{layer}=[lightgray,thick,fill=gray!30,rounded corners=20pt]
		\tikzstyle{dim}=[draw=black!50,pattern color=black!50];

    \def\DoubleDistance{45pt}

    \pgfmathsetlengthmacro\dsRadius{\DoubleDistance/2}

    \newcommand*{\dsPrev}[1]{\the\numexpr(#1)-1\relax}

    \begin{scope}[canvas is yx plane at z=1, transform shape]
        \path
            \foreach [
                var=\point,
                count=\i,
            ] in {
                (-3.5,0.5), (-3,2.5), (-1,3.5), (1.5,3), (4,3.5), (5,2.5), (5,0.5),
          (2.5,-2), (0,-0.5), (-3,-2), (-3.5,0.5)
            } {
                \point coordinate (dsM\i)
            }
            \pgfextra{\global\let\dsN\i}%
            %
            ($(dsM1)!\dsRadius!-90:(dsM2)$) coordinate (dsM1r)
            ($(dsM1)!\dsRadius!90:(dsM2)$) coordinate (dsM1l)
            ($(dsM\dsN)!\dsRadius!90:(dsM\dsPrev\dsN)$) coordinate (dsM\dsN r)
            ($(dsM\dsN)!\dsRadius!-90:(dsM\dsPrev\dsN)$) coordinate (dsM\dsN l)
            \foreach \i in {3, ..., \dsN} {
                    let \n2 = {\i},
                            \n1 = {\dsPrev\i},
                            \n0 = {\dsPrev{\dsPrev\i}},
                            \p0 = (dsM\n0),
                            \p1 = (dsM\n1),
                            \p2 = (dsM\n2),
                            \n{a} = {(atan2(\y2-\y1, \x2-\x1) - atan2(\y1-\y0, \x1-\x0))/2}
                    in
                            ($(\p1)!\dsRadius!-90-\n{a}:(\p2)$) coordinate (dsM\n1r)
                            ($(\p1)!\dsRadius!90-\n{a}:(\p2)$) coordinate (dsM\n1l)
            }
        ;
        \filldraw[
              fill=black!10,
              draw=black,
              very thin,
              variable=\i,
              smooth,
              tension=.7
        ]
        plot[
            samples at={1, ..., \dsN},
            mark options=black,
        ] (dsM\i)
            -- cycle
        ;

        \def\pts{(-1, 1), (-1.7, 1.9), (-1.9, -1), (-2.5, 0.3), (0.5, 1.2),
            (1.5, -0.8), (1.9, 2), (2.5, .4), (3.5, 1.6), (4, 0)};

        \foreach [var=\point] in \pts {
            \node at \point {};
        };

        \foreach [var=\px] in \pts {
            \foreach [var=\py] in \pts {
                \draw[thin,darkgray,opacity=.2] \px -- \py;
            };
        };

        \node[draw=none] (f) at (-1,1) {};

    \end{scope}

    \begin{scope}[canvas is yx plane at z=-1, transform shape]
        \path
            \foreach [
                var=\point,
                count=\i,
            ] in {
                (-3.5,0.5), (-3,2.5), (-1,3.5), (1.5,3), (4,3.5), (5,2.5), (5,0.5),
          (2.5,-2), (0,-0.5), (-3,-2), (-3.5,0.5)
            } {
                \point coordinate (dsM\i)
            }
            \pgfextra{\global\let\dsN\i}%
            %
            ($(dsM1)!\dsRadius!-90:(dsM2)$) coordinate (dsM1r)
            ($(dsM1)!\dsRadius!90:(dsM2)$) coordinate (dsM1l)
            ($(dsM\dsN)!\dsRadius!90:(dsM\dsPrev\dsN)$) coordinate (dsM\dsN r)
            ($(dsM\dsN)!\dsRadius!-90:(dsM\dsPrev\dsN)$) coordinate (dsM\dsN l)
            \foreach \i in {3, ..., \dsN} {
                    let \n2 = {\i},
                            \n1 = {\dsPrev\i},
                            \n0 = {\dsPrev{\dsPrev\i}},
                            \p0 = (dsM\n0),
                            \p1 = (dsM\n1),
                            \p2 = (dsM\n2),
                            \n{a} = {(atan2(\y2-\y1, \x2-\x1) - atan2(\y1-\y0, \x1-\x0))/2}
                    in
                            ($(\p1)!\dsRadius!-90-\n{a}:(\p2)$) coordinate (dsM\n1r)
                            ($(\p1)!\dsRadius!90-\n{a}:(\p2)$) coordinate (dsM\n1l)
            }
        ;
        \clip[variable=\i, smooth, tension=.7]
            plot[samples at={1, ..., \dsN}] (dsM\i)
            -- cycle
        ;

        \def\pts{(-1, 1), (-1.7, 1.9), (-1.9, -1), (-2.5, 0.3), (0.5, 1.2), (1.5,
            -0.8), (1.9, 2), (2.5, .4), (3.5, 1.6), (4, 0)};

        \xintForpair #1#2 in \pts \do{
            \edef\pta{#1,#2}
            \begin{scope}
                \xintForpair \#3#4 in \pts \do{
                    \edef\ptb{#3,#4}
                    \ifx\pta\ptb\relax 
                        \tikzstyle{myclip}=[];
                    \else
                        \tikzstyle{myclip}=[clip];
                    \fi;
                    \path[myclip] (#3,#4) to[half plane] (#1,#2);
                };
                \clip (-\maxxy,-\maxxy) rectangle (\maxxy,\maxxy); 
                \pgfmathsetmacro{\randhue}{rnd};
                \definecolor{randcolor}{hsb}{\randhue,.5,1};
                \fill[randcolor,opacity=.2] (#1,#2) circle (4*\biglen); 
            \end{scope}
        };

        \node[fac] (f1) at (-1.2,1.2) {};
        \node[fac] (f2) at (-.8,1.4) {};
        \node[fac] (f3) at (-.7,1) {};
        \node[fac] (f4) at (-.65,.6) {};
        \node[fac] (f5) at (-.55,.2) {};
        \node[fac] (f6) at (-1.1,.4) {};
        \node[fac] (f7) at (-1,.75) {};

    \end{scope}

    \draw[lightgray,thin,dashed] (f) -- (f1);
    \draw[lightgray,thin,dashed] (f) -- (f2);
    \draw[lightgray,thin,dashed] (f) -- (f3);
    \draw[lightgray,thin,dashed] (f) -- (f4);
    \draw[lightgray,thin,dashed] (f) -- (f5);
    \draw[lightgray,thin,dashed] (f) -- (f6);
    \draw[lightgray,thin,dashed] (f) -- (f7);

\end{tikzpicture}
        \fi
        \caption{An \(\ell\)-centered instance. In the upper layer there is a set
			\(S\) of \(l\) vertices connected as a clique. The rest of vertices
			are divided into separate \mar{cluster}s. Vertices in a single \mar{cluster} are only
			connected to their center in the set \(S\).}\label{fig:l-centered}
	\end{figure}
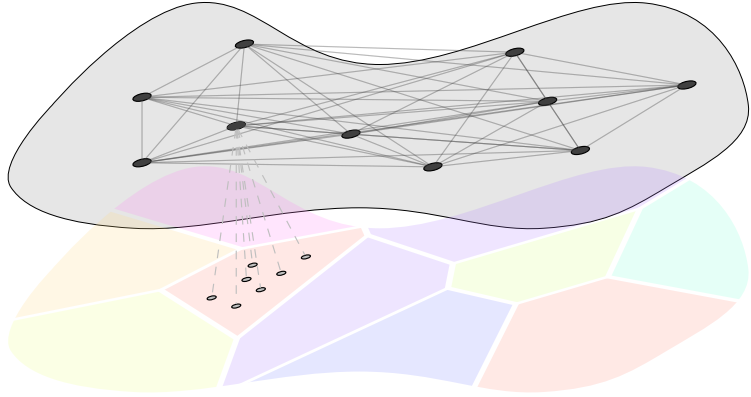

For a center \(s\in S\) we shall say that all
	nodes from $F \cup C$ that are connected
	to \(s\) form a \emph{cluster} of \(s\).
    If we consider only nodes from $F$,
    then \mar{we talk about}
 an $f$-cluster of $s$, denoted $F(s)$.

	\janr{Skoro mamy \(F(s)\), to moze sie jeszcze szarpniemy na \(C(s)\);--- tylko potrzebujesz $C(s)$ w ktorymkolwiek momencie to mozna, jak najbardziej, jeszcze jak }

	In the following lemma we relate the cost of embedding the optimum
	solution \(OPT\br{d}\) from a metric \(d\) to \(d_{l}\).

	\begin{lemma}[Embedding \(d\) into \(l\)-centered metric \(d_{l}\)]\label{lem:embedding-to-dl}
 		Let $\algunc$ be a solution for the \UKM problem on metric $d$ from which we construct the $\ell$-centered instance.
        Optimal solution \(OPT\br{d}\) can be embedded into an \(l\)-centered
		metric \(d_{l}\) with the cost relation being
		\[
			\cost{\phi^{OPT\br{d}}}{d} \leqslant \cost{\phi^{OPT\br{d}}}{d_{l}}
			\leqslant 3 \cdot \cost{\phi^{OPT\br d}}{d}
			+ 4 \cdot \cost{\phi^{\algunc}}{d}.
		\]
	\end{lemma}
	\begin{proof}
\marr{Lemma 2 shouldn't be done for $OPT(d)$ but rather for any lagorithm for metric $d$; thanks to this we will not have to write Lemma 3}
		Let \(c\) be a client connected to facility \(f_{c}\) in the optimal
		solution \(OPT\br{d}\). Let \(s^c\) be the center closest to \(c\)
		within \(S\) (the \(\ell \)-center), and let \(s^{f_c}\) be the center
		closest to \(f_{c}\). First let us note that \(d_{l}\br{c,f_{c}} =
		d\br{c,s^c} + d\br{s^c,s^{f_c}} + d\br{s^{f_c},f_{c}}\). Now let us
		bound the terms \(d\br{f_{c},s^{f_c}}\) and \(d\br{s^c,s^{f_c}}\)
		separately.

		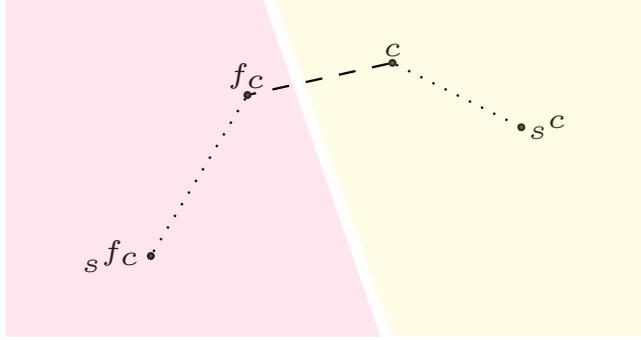
\begin{figure}[h]
			\centering
\ifdefined\STANDALONE{}
\def\biglen{20cm}
\tikzset{
	half plane/.style={ to path={
			 ($(\tikztostart)!.5!(\tikztotarget)!#1!(\tikztotarget)!\biglen!90:(\tikztotarget)$)
		-- ($(\tikztostart)!.5!(\tikztotarget)!#1!(\tikztotarget)!\biglen!-90:(\tikztotarget)$)
		-- ([turn]0,2*\biglen) -- ([turn]0,2*\biglen) -- cycle}},
	half plane/.default={1pt}
}

\begin{tikzpicture}[
		every node/.style={draw,circle,fill=darkgray,inner sep=.3pt},
		every label/.append style={rectangle, font=\tiny}]

	\tikzstyle{lab}=[rectangle,draw=none,fill=none,inner sep=0,outer sep=0,font=\tiny];

	\node[label={right:\(s^c\)}] (sc) at (2, 1) {};
	\node[label={left:\(s^{f_c}\)}] (sf) at (-.3, .2) {};

	\node[label={above:\(f_c\)}] (f) at (.3,1.2) {};
	\node[label={above:\(c\)}] (c) at (1.2,1.4) {};

	\clip (-1.2, -.3) rectangle (2.8, 1.8);
	\begin{scope}
		\path[clip] (sc) to[half plane] (sf);
		\pgfmathsetmacro{\randhue}{rnd};
        \definecolor{randcolor}{hsb}{\randhue,.5,1};
		\fill[randcolor,opacity=.2] (sf) circle (4*\biglen);
	\end{scope}
	\begin{scope}
		\path[clip] (sf) to[half plane] (sc);
		\pgfmathsetmacro{\randhue}{rnd};
        \definecolor{randcolor}{hsb}{\randhue,.5,1};
		\fill[randcolor,opacity=.2] (sc) circle (4*\biglen);
	\end{scope}

	\draw[dashed] (c) -- (f);
	\draw[dotted] (c) -- (sc);
	\draw[dotted] (f) -- (sf);

\end{tikzpicture}
\fi
			\caption{Situation in Lemma~\ref{lem:embedding-to-dl}.
            In the optimum solution to the \CKM{} instance, client $c$ is
				connected to the facility \(f_c\). In the \(\ell\)-centered
				instance $c$ resides in a cell, where \(s^c\) is a center. The
				center of \(f_c\) is \(s^{f_c}\).}\label{fig:ineq}

		\end{figure}
                \marr{can we make to superscript $f_c$ in $s^{f_c}$ smaller?; tried small and tiny but didn't work}
		\begin{fact}\label{fact:fcfprimc}
			For every client \(c\) and its facility \(f_c\) from \(OPT\) we have
			\(d\br{f_{c},s^{f_{c}}} \leqslant d\br{f_{c},c}+d\br{c,s^c}\).
		\end{fact}
		\begin{proof}
			Since \(s^{f_c}\) is the closest \(\ell\)-center to the facility
			\(f_c\), we have that \(d\br{f_{c},s^{f_{c}}} \leqslant
			d\br{f_{c},s^c}\). At the same time, from the triangle inequality it
			follows that \(d\br{f_c, s^c} \leqslant d\br{f_c, c} + d\br{c, s^c}
			\).
		\end{proof}
		\begin{fact}\label{fact:cprimfprimc}
			For each \(c\) we have \(d\br{s^c, s^{f_c}} \leqslant
			2\br{d\br{f_{c},c} + d\br{c,s^c}}\).
		\end{fact}
		\begin{proof}

			From the triangle inequality we know that
			\[
				d\br{s^c,s^{f_{c}}} \leqslant d\br{s^c,c} + d\br{c,f_{c}}
				+ d\br{f_{c},s^{f_{c}}}.
			\]
			From Fact~\ref{fact:fcfprimc} we also know that
			\(d\br{f_{c},s^{f_{c}}} \leqslant d\br{f_{c},c}+d\br{c,s^c}\), and
			combining the two inequalities we get \(d\br{s^c,s^{f_{c}}}
			\leqslant d\br{s^c,c} + d\br{c,f_{c}} + d\br{f_{c},s^{f_{c}}}
			\leqslant 2\br{d\br{f_{c},c}+d\br{c,s^c}}\).

		\end{proof}
		These facts imply
		\begin{eqnarray*}
			d_{l}\br{c,f_{c}} & =
				& d\br{c,s^c} + d\br{s^c,s^{f_{c}}} + d\br{s^{f_{c}},f_{c}} \\
				& \leqslant & d\br{c,s^c} + d\br{s^c,s^{f_{c}}} + \br{d\br{f_{c},c} + d\br{c,s^c}} \mbox{\qquad (from Fact~\ref{fact:fcfprimc})} \\
				& \leqslant & d\br{c,s^c} + 2\br{d\br{f_{c},c} + d\br{c,s^c}} + \br{d\br{f_{c},c}+d\br{c,s^c}} \qquad\mbox{(from Fact~\ref{fact:cprimfprimc})} \\
				& = & 3\cdot d\br{f_{c},c}+4\cdot d\br{c,s^c},
		\end{eqnarray*}
		which implies the second inequality from the statement of
		Lemma~\ref{lem:embedding-to-dl}. The first one directly comes from the
		triangle inequality
        \[
        	d\br{c,f_{c}} \leqslant d\br{c,s^c} + d\br{s^c, s^{f_{c}}} +
        	d\br{s^{f_{c}},f_{c}} = d_{l}\br{c,f_{c}},
        \]
        completing the whole proof.

	\end{proof}

Another lemma is quite simple. Its proof just comes from the fact
that metric $d_{l}$ dominates the metric $d$, i.e., $d_{l}\br{u,v}\geqslant d\br{u,v}$
for all pairs of vertices $u,v \in C \cup F$.

\begin{lem}[Going back from $l$-centered metric $d_{l}$ to $d$]\label{lem:back-to-d}
 Any solution for the $l$-centered metric $d_{l}$ can be embedded
back into $d$ without any loss:
\[
\cost{\phi^{ALG\br{d_{l}}}}{d_{l}}\geqslant\cost{\phi^{ALG\br{d_{l}}}}d.
\]
\end{lem}

Blending together Lemmas~\ref{lem:embedding-to-dl} and~\ref{lem:back-to-d} we can state the following Lemma about reducing the \CKM problem to $\ell$-centered instances.

\begin{lem}\label{lem:reduction-to-lcenter}
Suppose we are given a solution \mar{$\algunc$} for the \UKM problem on metric $d$ which opens $\ell$ centers, but $\beta$-approximates the optimum solution \mar{$OPT^k_{unc}\br{d}$} for \UKM problem with $k$ centers\mar{, i.e., $\cost{\algunc}{d} \leqslant \beta \cdot \cost{OPT^k_{unc}\br{d}}{d}$}.
Suppose we are given an $\alpha$-approximation algorithm for the \CKM problem on $\ell$-centered instances. If so, then we can construct an $\alpha\cdot\br{3+4\beta}$-approximation algorithm for \CKM on general instances.
\end{lem}

\begin{proof}
Suppose that we have an $\alpha$-approximation \mar{solution} for
the $\ell$-centered instance with metric $d_l$, i.e., $ALG(d_{l})$
such that
\[
\cost{\phi^{ALG(d_{l})}}{d_{l}} \leqslant\alpha\cdot \cost{\phi^{OPT\br{d_{l}}}}{d_l}.
\]
Since $OPT(d)$ is some solution for the $\ell$-centered instance with metric $d_l$ we have
\[
\cost{\phi^{ALG(d_{l})}}{d_{l}}
\leqslant\alpha\cdot \cost{\phi^{OPT\br{d_{l}}}}{d_l}
\leqslant\alpha\cdot \cost{\phi^{OPT\br{d}}}{d_l}.
\]
And from Lemma~2 we have that
\begin{align*}
\cost{\phi^{ALG(d_{l})}}{d_{l}}
&\leqslant  \alpha\cdot \cost{\phi^{OPT\br{d_{l}}}}{d_l} \\
&\leqslant  \alpha\cdot \cost{\phi^{OPT\br{d}}}{d_l}\\
&\leqslant  \alpha \br{3\cdot\cost{\phi^{OPT\br d}}d+4\cdot \cost{\phi^{\algunc}}{d}}.
\end{align*}
Since solution $\algunc$ $\beta$-approximates the optimal solution $OPT^k_{unc}\br{d}$ for \UKM with $k$ centers on metric $d$, we have that
$$
\cost{\phi^{\algunc}}{d} \leqslant \beta \cdot \cost{\phi^{OPT^k_{unc}\br{d}}}{d} \leqslant \beta  \cdot \cost{\phi^{OPT\br{d}}}{d}.
$$
The second inequality $\cost{\phi^{OPT^k_{unc}\br{d}}}{d} \leqslant \cost{\phi^{OPT\br{d}}}{d}$ follows from an obvious fact that uncapacitated version of the problem is easier than the capacitated.
Hence
\begin{align*}
\cost{\phi^{ALG(d_{l})}}{d_{l}}
&\leqslant \alpha \br{3\cdot\cost{\phi^{OPT\br d}}d+4\cdot \cost{\phi^{\algunc}}{d}}\\
&\leqslant \alpha \br{3\cdot\cost{\phi^{OPT\br d}}d+4\beta\cdot \cost{\phi^{OPT\br{d}}}{d}}\\
&\leqslant \alpha\br{3+4\beta}\cdot\cost{\phi^{OPT\br d}}d.
\end{align*}
Since without any loss we can embed the solution $ALG(d_{l})$ for the $\ell$-centered metric $d_l$ into the initial metric $d$
 (Lemma~3) we obtain
an $\alpha\cdot\br{3+4\beta}$-approximation algorithm.
The claim follows.\end{proof}

\section{Constant factor approximation}
In this section we present the main result of the paper which is a $\br{7+\eps}$-approximation algorithm for the \NUCKM problem. We precede it with a $\br{7+\eps}$-approximation algorithm for the \UCKM problem to introduce the ideas gradually.
Both algorithms enumerate configurations of open facilities' locations, and as a subroutine we need to use an algorithm which, for a fixed configuration of $k$ open facilities, finds the optimal assignment of clients to facilities. This subroutine is presented in the following subsection.

\subsection{Optimal mapping subroutine}
We are given an $\ell$-centered metric instance $(F \cup C \cup S, d_{\ell})$ of the $k$-median problem.
\mic{Suppose that we have already decided to open a fixed subset $F^{open}\subseteq F$ of the facilities 
and we look for a mapping $\phi:C \to F^{open}$.
In the uncapacitated case we can just assign each client to the closest facility in $F^{open}$.
It turns out that even in the capacitated setting  we can find the mapping $\phi$ optimally in polynomial time for a given $F^{open}$.}
We state the problem of finding the optimal $\phi$ as an integer program:
\begin{align*}
\mbox{\rm minimize }
 \quad \sum_{c \in C} \sum_{f \in F^{open}} &d_\ell(c, f)\cdot x_{c,f} & &&\mbox{(MAPPING-IP)}\\
\mbox{subject to}  \quad\qquad \sum_{f \in F^{open} } &x_{c,f} = 1 \qquad &\forall c \in C, \\
 \quad \sum_{c \in C}\ &x_{c,f} \leqslant u_f \qquad &\forall f \in F^{open}, \\
 \quad &x_{c,f} \in \{0,1\}.&&
\end{align*}
In the above program $x_{c,f}=1$ represents the fact that $\phi(c)=f$.

\begin{lem}\label{lem:optimal-mapping}
We can find an optimal solution to the \emph{(MAPPING-IP)} in polynomial time.
\end{lem}
\begin{proof}
The proof follows from the fact that the relaxation of the above integer program --- a program which differs from (MAPPING-IP) only with the $x_{c,f} \geqslant 0$ constraints instead of $x_{c,f} \in \set{0,1}$ ---
has an optimal solution which is integral.
\mic{To see this, observe} that the linear program is a formulation of the transportation problem. 
For such a \mic{linear program,} the constraint matrix is totally unimodular, which implies the integrality of an extremal solution. See~\cite{schrijver-book} for a reference.
\end{proof}
\janr{Czy argument, że to jest po prostu pełny graf dwudzielny, i dlatego to jest TU, nie jest elegantszy?}
\marr{not exactly, because in the literature the TU argumentation for the incidence matrix of a bipartite graph always is done when one considers an LP of the form $\min cx \mbox{s.t. } Ax \geqslant b$; we have a program with upperbounds in the constraints as well, and only in the transportation problems context the integrality of TU for such problems is shown.}
\micr{or argument by MinCostMaxFlow?}

\subsection{Uniform case}\label{sec:uniform}
    As a warm up, we begin with a parameterized algorithm for the uniform case.
    It is a bit simpler than the general case, because once we know the number
    of facilities to open in $f$-cluster $F(s)$, then we can  choose them
    greedily.

\begin{lem}
\label{lem:fpt-uniform} \UCKM can be solved exactly in time $\ell^{k}\cdot n^{O(1)}$
on $\ell$-centered instances. \end{lem}\mesr{Use O* notation throughout the paper.}
\begin{proof}Let $(F \cup C \cup S, d_{\ell})$ be the $\ell$-centered metric.
Note that the $f$-clusters partition the whole set of facilities, i.e., $\cup_{s\in S} F(s) = F$.
Let $OPT\br{d_{\ell}}$ be an optimal solution for the \CKM problem on $d_{\ell}$.
Every facility $f\in F$ belongs to exactly one $f$-cluster $F(s)$.
Hence, the $f$-clusters partition the set of $k$ facilities opened by $OPT(d_{\ell})$.
Let us look at all the facilities from a particular $f$-cluster $F(s)$ opened by $OPT(d_{\ell})$, and suppose that $OPT(d_{\ell})$ opens $k_s$ of facilities in $F(s)$. Since we consider a uniform capacity case, we can assume without loss that these $k_s$ open facilities from  $F(s)$ are exactly the ones that are closest to $s$.

\marr{here $F(s)$ is called cluster but we mean only facilities; maybe we should just called it $f$-cluster}

Therefore, if we know what is the number of facilities that $OPT(d_{\ell})$ opens in each $f$-cluster, then we would know what the exact set of open facilities in $OPT\br{d_{\ell}}$ is due to the greediness in each $f$-cluster.
To find out this allocation we can simply enumerate over all possibilities. We just need to scan over all \emph{configurations} $\br{k_s}_{s\in S}$ where $\sum_{s}k_s = k$.
Since there are $k$ facilities to open, and each of them can belong to one of $\ell$ $f$-clusters $F(s)$, there are at most $\ell^k$ possible configurations. Of course some configurations may not be feasible since it may happen that $k_s > \size{F(s)}$, but these can be simply ignored.

For each configuration $\br{k_s}_{s\in S}$ we need to find the optimal mapping of clients to the set of open facilities that preserves their capacities.
Let $F\br{\br{k_s}_{s\in S}}$ be the set of open facilities induced by configuration $\br{k_s}_{s\in S}$, that is, where we greedily open $k_s$ facilities in $f$-cluster $F(s)$.
Given $F\br{\br{k_s}_{s\in S}}$, to find the optimal mapping we use the polynomial time exact algorithm from Lemma~\ref{lem:optimal-mapping} with $F^{open} = F\br{\br{k_s}_{s\in S}}$.

Once we know the optimal assignment for each configuration, we can simply take the cheapest one, knowing that it is the optimal one. This proves the lemma.
\end{proof}

\mic{This lemma suffices to obtain a $\br{7+\eps}$-approximation for \UCKM with a reasoning that we will present in Theorem~\ref{thm:nuckm} in full generality.}

\subsection{Non-uniform case}\label{sec:nonuniform}

\begin{lem}
\label{lem:fpt-non-uniform} \NUCKM can be solved with approximation
ratio $(1+\epsilon)$ in time \\ $\br{O\br{\ell\cdot\frac{1}{\eps}\ln\frac{n}{\eps}}}^{k}n^{O(1)}$
on $\ell$-centered instances.\end{lem}

\begin{proof}
We begin with guessing the largest distance in $d_\ell$ between a client and a facility that would appear in the optimal solution --- let us denote this quantity as $D$.
There are at most $O(n^2)$ choices for $D$, and from now we assume that it is guessed correctly.
Note that $D \le \cost{OPT(d_\ell)}{d_\ell}$ and $D \ge d(f,s_f)$ for all facilities opened by $OPT(d_\ell)$.

Consider the set of facilities $F\br s$ in the cluster of a center $s$.
We can remove all facilities $f$ such that $d(s,f) > D$, because they cannot be a part of the optimal solution.
Let us partition remaining facilities
from $F\br s$ into buckets $F_{0}\br s,F_{1}\br s,...,F_{\ceil{\log_{1+\eps}\frac{n}{\eps}}}\br s$,
such that
\[
F_{i}\br s=\begin{cases}
\setst{f\in F\br s}{d\br{s,f}\in\brq{\br{1+\eps}^{-\br{i+1}}D,\br{1+\eps}^{-i}D}} & \mbox{for }i<\ceil{\log_{1+\eps}\frac{n}{\eps}}\\
\\
\setst{f\in F\br s}{d\br{s,f}\in\brq{0,\br{1+\eps}^{-\ceil{\log_{1+\eps}\frac{n}{\eps}}}D}} & \mbox{for }i=\ceil{\log_{1+\eps}\frac{n}{\eps}}
\end{cases}
\]

The number of buckets equals $\log_{1+\eps}\frac{n}{\eps}=\frac{1}{\ln\br{1+\eps}}\ln\frac{n}{\eps}=O\br{\frac{1}{\eps}\ln\frac{n}{\eps}}$.
We modify the metric again by setting $d'_\ell(s,f) = \br{1+\eps}^{-i}D$ for $f \in F_{i}\br s$.
The distances within $S$ remain untouched.
Observe that the distances can only increase.

We shall guess the structure of the solution $OPT(d'_\ell)$ similarly as in Lemma~\ref{lem:fpt-uniform}.
For each of the $k$ facilities,
we can choose its location as follows: first we choose one of the
$\ell$-centers $s$ ($\ell$ choices), and then we choose one of the
$F_{i}\br s$ partitions ($O\br{\frac{1}{\eps}\ln\frac{n}{\eps}}$
choices).
Let us denote the number of facilities in a particular partition
$F_{i}\br s$ as $k_{s,i}$.
We can assume that $k_{s,i} \le \size{ F_{i}\br s}$
because otherwise we know that the guess was incorrect.
Since $d'_\ell(s,f)$ is the same for all $f \in F_{i}\br s$, we can assume the optimal solution opens $k_{s,i}$
facilities with the biggest
capacities.

Once we establish the set of facilities to open,
we can find the optimal assignment in metric $d'_\ell$ using the polynomial time exact subroutine from Lemma~\ref{lem:optimal-mapping}.
The total time complexity of solving the problem exactly over $d'_\ell$ equals the running time of the subroutine
times the number of possible configurations, which is $\br{O\br{\ell\cdot\frac{1}{\eps}\ln\frac{n}{\eps}}}^{k}n^{O(1)}$.

\mic{
It remains to prove that the algorithm yields a proper approximation.
We will show that for any solution $SOL$ it holds that
\begin{equation}\label{eq:non-uniform}
\cost{\phi^{SOL}}{d_\ell} \leqslant \cost{\phi^{SOL}}{d'_\ell} \leqslant
			\br{1+\eps}\cdot \cost{\phi^{SOL}}{d_\ell} + \eps\cdot D.
\end{equation}
By substituting $SOL=OPT(d_\ell)$ we learn that there exists a solution over metric $d'_\ell$ of cost at most $\br{1+\eps}\cdot \cost{\phi^{OPT(d_\ell)}}{d_\ell} + \eps\cdot D \le \br{1+2\eps}\cdot \cost{\phi^{OPT(d_\ell)}}{d_\ell}$ for correctly guessed~$D$.
Therefore the cost of the solution found by our algorithm cannot be larger.
Finally we substitute this solution as $SOL$ to see that
its cost cannot increase when returning to metric $d_\ell$.
The claim will follow by adjusting $\eps$.

The first inequality in (\ref{eq:non-uniform}) is straightforward because $d'_\ell$ dominates $d_\ell$.
Consider now a pair $(c,f=\phi^{SOL}(c))$, where $f \in F_{i}\br s$.
If $i<\ceil{\log_{1+\eps}\frac{n}{\eps}}$, then $d_\ell(c,f) \le d'_\ell(c,f) \le (1+\eps)\cdot d_\ell(c,f)$, so the cost of connecting such pairs increases at most by a multiplicative factor $(1+\eps)$ during the metric switch.
If $i = \ceil{\log_{1+\eps}\frac{n}{\eps}}$,
then $d'_\ell(s,f) = \frac{\eps D}{n}$.
Since there are at most $n$ such pairs, the total additive cost increase is bounded by $\eps\cdot D$.
}

\end{proof}

\begin{thm}\label{thm:nuckm}
\NUCKM can be solved with approximation ratio $(7+\epsilon)$ in
time $(k/\epsilon)^{O(k)}n^{O(1)}$. \end{thm}
\begin{proof}
From Lemma~\ref{lem:fpt-non-uniform} we know that we can get a $\br{1+\eps}-$approximation
algorithm for the \NUCKM problem on $\ell$-centered instances in
time $\br{O\br{\ell\cdot\frac{1}{\eps}\ln\frac{n}{\eps}}}^{k}n^{O(1)}$.
We shall use the $\br{1+\eps}$-approximation for \UKM by Lin and Vitter~\cite{lin92}, that opens at most $\ell=\br{1+\frac{1}{\eps}}k\cdot (\ln n + 1)$ facilities. By plugging this subroutine to find $\ell$-centers
into the Lemma~\ref{lem:reduction-to-lcenter} together with Lemma~\ref{lem:fpt-non-uniform}, we obtain a $\br{7+\eps}-$approximation
algorithm for the general \NUCKM problem with running time
$$O\br{\br{\br{1+\frac{1}{\eps}}k\cdot(\ln n+1)\cdot\frac{1}{\eps}\ln\frac{n}{\eps}}^{k}}n^{O(1)}=O\br{\br{\frac{1}{\eps^{O(1)}}k\ln^{2}n}^{k}}n^{O(1)}.$$
Finally, we use standard arguments to show that $\br{\ln n}^{2k} \le \max(n, k^{O(k)})$. Consider two cases. If $\frac{\ln n}{2\ln\ln n}\le k,$
then by inverting we know that $\ln n=O\br{k\ln k}$, and so $\br{\ln n}^{2k}=k^{O(k)}$. Suppose now that $\frac{\ln n}{2\ln\ln n}>k$.
In this case
\[
\br{\ln n}^{2k}<\br{\ln n}^{\frac{\ln n}{\ln\ln n}}=2^{\ln\ln n\cdot\frac{\ln n}{\ln\ln n}}=n.
\]
\end{proof}

\section{\(\OO\br{\log k}\)-approximation}\label{sec:folklore}
\micr{moved this section to the end}
    In this section present a polynomial-time \(\OO\br{\log
    k}\)-approximation algorithm for \CKM{}.
    A constant-factor approximation
    algorithms for \UKM{} exist~\cite{charikar1999constant}, and so it is a
    clear consequence of the Lemma~\ref{lem:reduction-to-lcenter} with \(\beta
    \) being constant that it is sufficient for us to construct an \(\OO\br{\log
    k}\)-approximation algorithm for the \(k\)-centered instances.

    A standard tool to provide such a guarantee is the \emph{Probabilistic Tree Embedding} by~\cite{FakcharoenpholRT03}. \mar{This makes our algorithm a randomized one, but if needed, it is possible to derandomize it using the ideas from~\cite{DBLP:conf/focs/CharikarCGGP98}.}
    \begin{definition}\label{def:tree-embedding}
        A set of metric spaces \(\mathcal{T}\) together with a probability distribution $\pi_{\mathcal{T}}$ over $\mathcal{T}$ probabilistically
        \(\alpha\)-approximates the metric space \((X, d)\) if
        \begin{enumerate}
            \item Every metric \(\tau \in \mathcal{T}\) dominates \((X, d)\),
                that is, \(d(x, y) \leqslant \tau(x, y)\ \forall x,y \in X\).
            \item For every pair of points \(x, y \in X\) its expected distance
                is not expanded by more then \(\alpha\), i.e.,
                \[
                    \mathbb{E}_{\tau \sim \pi_{\mathcal{T}}}[\tau(x, y)] \leqslant
                    \alpha \cdot d(x, y).
                \]
        \end{enumerate}
    \end{definition}
    It is a well-known fact, that any metric \((X, d)\), can be
    probabilistically \(\OO(\log |X|)\)-approximated by a distribution of tree
    metrics, such that the points in \(X\) are the leaves in the resulting
    tree~\cite{FakcharoenpholRT03}.

    As described in Definition~\ref{def:l-centered}, our \(k\)-centered metric
    \(d_k\) is induced by a graph composed of two layers --- the set \(S\) of
    \(k\) vertices connected in a clique, and the rest of vertices, \(F \cup
    C\), each connected to only one vertex in \(S\). Let \(T\) be a random tree
    embedding of the set \(S\) (with a metric function \(d_T\)). A modified
    instance \(G_T\) of our problem is created by replacing the clique \(S\)
    with its tree approximation \(T\).

    \begin{lemma}\label{lem:tree-opt}
        An optimum solution for \CKM{} on the instance \(G_T\) is in expectation
        at most \(\OO\br{\log k}\) times larger than the optimum for the metric
        \(d_k\).
    \end{lemma}
    \begin{proof}

        \(OPT\br{d_{k}}\) denotes the optimum mapping of clients to \(k\)
        facilities in the \(k\)-centered metric \(d_k\). Consider client \(c\)
        and facility \(f = \phi^{OPT\br{d_{k}}}\br{c}\). 
        Let now \(s^c\) be the center of \(c\) and \(s^f\) the center
        of \(f\). The cost of connecting client \(c\) to \(f\) amounts to
        \[
            d_k\br{c, f} = d_k\br{c, s^c} + d_k\br{s^c, s^f} + d_k\br{s^f, f}
        \]
        in the metric \(d_k\).

        The guarantee of tree embeddings gives us an upper bound on a cost of
        applying the same mapping in the instance \(G_T\),
        \begin{eqnarray*}
            \mathbb{E}\brq{d_T\br{c, f}} & =
                & d_k\br{c, s^c} + \mathbb{E}\brq{d_T\br{s^c, s^f}} + d_k\br{s^f, f} \\
                & \leqslant & d_k\br{c, s^c} + \OO\br{\log k} \cdot d_k\br{s^c, s^f}  + d_k\br{s^f, f} \\
                & \leqslant & \OO\br{\log k} \cdot d_k\br{c, f}.
        \end{eqnarray*}

        Which means that \(\mathbb{E}\brq{cost\br{\phi^{OPT\br{d_{k}}}, d_T}}
        \leqslant \OO\br{\log k} \cdot cost\br{\phi^{OPT\br{d_{k}}}, d_{G_k}}\).
        Moreover, \(OPT\br{d_{k}}\) might not be the optimal solution for the
        metric \(d_{T}\), yet its optimal solution can only have
        smaller cost:
        \[
            cost\br{\phi^{OPT\br{d_{T}}}, d_T} \leqslant
            cost\br{\phi^{OPT\br{d_{k}}}, d_T}
        \]
    \end{proof}
    
    \begin{thm}
        The CKM problem admits an \(\OO(\log k)\)-approximation algorithm with
        polynomial running time.
    \end{thm}

    \begin{proof}
        After applying the probabilistic tree embedding to the graph inducing
        \(d_k\) --- as presented in Lemma~\ref{lem:tree-opt} --- we obtain a
        tree instance \(G_T\).
\mic{It should come as no surprise that the problem is polynomially solvable on trees and we explain how to find the optimum solution on \(G_T\) in Lemma~\ref{lem:ckm-tree}}.
       The assignment
        \(\phi^{OPT(d_T)}\), which yields the minimum cost on the tree \(G_T\),
        can be now used to match clients to facilities in the original instance.
        It does not incur any additional cost, as
        \[
            cost\br{\phi^{OPT(d_T)}, d_T} \geqslant
            cost\br{\phi^{OPT(d_T)}, d_{k}}
            \geqslant
            cost\br{\phi^{OPT(d_T)}, d}
        \]
        \sloppy from the property (1) of Definition~\ref{def:tree-embedding} and
        Lemma~\ref{lem:back-to-d}. Combining this with a bound on
        \(\mathbb{E}\brq{cost\br{\phi^{OPT\br{d_{k}}}, d_T}}\) from
        Lemma~\ref{lem:tree-opt} finishes the proof.
\end{proof}

    \subsection{\CKM{} on a tree}
    The second ingredient to the \(\OO(\log k)\)-approximation for \CKM{} is
    solving the problem exactly on trees. We will now describe a simple, exact,
    polynomial algorithm for that special case. In our algorithm we can assume,
    that all the clients and facilities reside in leaves, but the principle is
    easy to extend to the general problem on trees.

    Imagine we have a subtree \(t\) of the tree instance, hanging on an edge
    \(e_t\). Once we have decided, which facilities to open inside the subtree
    \(t\), we know if their total capacity is sufficient to serve all the
    clients inside \(t\). If not, then we need to route some clients'
    connections to the facilities outside through the edge \(e_t\). However, if
    the facilities we have opened in \(t\) have enough total capacity to serve
    some \(b\) clients from the outside, we will connect them through the edge
    \(e_t\) (see Figure~\ref{fig:folk-subtree}).
    \begin{figure}
        \centering{}
        \begin{minipage}[b]{.33\linewidth}
            \centering{}
            \ifdefined\STANDALONE{}
            	\begin{tikzpicture}[
        every node/.style={draw=darkgray!70,circle,fill=darkgray!70,inner sep=.3pt},
        every label/.append style={rectangle,text=black, font=\tiny}]

    \tikzstyle{ghost}=[draw=none,fill=none,outer sep=0,inner sep=0,text=darkgray!80];
    \tikzstyle{cli}=[darkgray!70,circle,inner sep=1pt];
    \tikzstyle{fac}=[darkgray!70,rectangle,inner sep=1.2pt];
    \node[ghost] (r) at (0, 0) {};
    \node[ghost] (l1) at (-1, -2) {};
    \node[ghost] (l2) at (1, -2) {};

    \fill[draw=lightgray,thick,fill=lightgray!30] (r.center) -- (l1.center) -- (l2.center)
        -- cycle;

    \node (p) at (-.2,1) {}; \draw[thick, lightgray] (p) --
        node[midway,right,ghost] {\(e_t\)} (r.center);

    \node[cli] at (-.78, -1.9) {};
    \node[cli] at (-.63, -1.9) {};
    \node[cli] at (.12,  -1.9) {};
    \node[cli] at (.87,  -1.9) {};

    \node[fac] at (-.46, -1.9) {};
    \node[fac] at (-.1,  -1.9) {};
    \node[fac] at (.6,   -1.9) {};

    \node[ghost] at (0,-1.2) {\(t\)};

\end{tikzpicture}
			\fi
            \subcaption{}
        \end{minipage}
        \begin{minipage}[b]{.64\linewidth}
            \centering{}
            \ifdefined\STANDALONE{}
           		\tikzset{every picture/.style={draw=lightgray,thick}}
\begin{tikzpicture}[
        every node/.style={draw=darkgray!70,circle,fill=darkgray!70,inner sep=.3pt},
        every label/.append style={rectangle,text=black, font=\large}]

    \tikzstyle{ghost}=[draw=none,fill=none,outer sep=0,inner sep=0,text=darkgray!70];
    \tikzstyle{cli}=[darkgray!70,circle,inner sep=1pt];
    \tikzstyle{fac}=[darkgray!70,rectangle,inner sep=1.2pt];

    \node (r) at (0,0) {};
    \node[ghost] (r1) at (-1.5, -1) {};
    \node[ghost] (l11) at (-2,-2) {};
    \node[ghost] (l12) at (-1, -2) {};

    \node[ghost] (r2) at (1.5, -1) {};
    \node[ghost] (l21) at (1, -2) {};
    \node[ghost] (l22) at (2, -2) {};

    \draw[fill=lightgray!30] (r1.center) -- (l11.center) -- (l12.center)
    -- cycle;
    \node[ghost] at (-1.5,-1.7) {\(t_1\)};
    \draw[fill=lightgray!30] (r2.center) -- (l21.center) -- (l22.center)
    -- cycle;
    \node[ghost] at (1.5,-1.7) {\(t_2\)};

    \draw (r) -- node[midway,left,ghost] {\(e_{t_1}\)} (r1.center);
    \draw (r) -- node[midway,right,ghost] {\(e_{t_2}\)} (r2.center);

    \node (p) at (-.2,1) {}; \draw[thick, lightgray] (p) --
        node[midway,right,ghost] {\(e_t\)} (r.center);

    \node[thin,draw=none,rectangle,fill=none,font=\tiny] at (0,-.3) {\(\langle k', b\rangle\)};
    \node[thin,draw=none,rectangle,fill=none,font=\tiny] at (-1,-1.1) {\(\langle k'_1, b_1\rangle\)};
    \node[thin,draw=none,rectangle,fill=none,font=\tiny] at (1,-1.1) {\(\langle k'_2, b_2\rangle\)};
\end{tikzpicture}
            \fi
            \subcaption{}
        \end{minipage}
        \caption{(a) A subtree \(t\) with some open facilities (squares) and a
            number of clients (circles). If the total capacity of facilities
            opened in \(t\) exceeds the needs of the clients, we may decide to
            connect some clients from outside through edge \(e_t\). If the
            capacity of open clients is too small, we may connect some clients
            with outside facilities through \(e_t\). It never makes sense to do
            both. (b)~A dynamic programming step. To compute \(D(t, k', b)\) we
            need to find the cheapest solutions for subtrees, such that \(k'_1 +
            k'_2 = k'\) and \(b_1 + b_2 = b\). The additional cost incurred is
            \(d(e_t)\cdot |b| \).}\label{fig:folk-subtree}
    \end{figure}
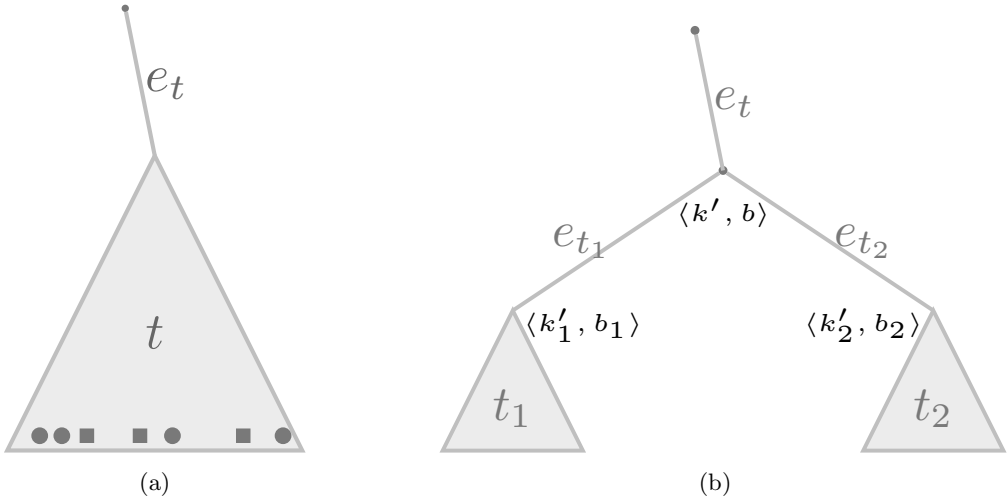
    This insight lays out the dynamic algorithm for us.
\mes{We first turn the tree into a complete binary tree by adding dummy vertices and edges of length \(0\) (which may double its size).}
Then, for every subtree \(t\), numbers \(k'\) and \(b\), we
    compute \(D(t, k', b)\).
    \begin{definition}
        \(D(t, k', b)\), for subtree \(t\), number \(k' \in \{0, \dots, k\} \)
        of facilities and balance \(b \in \{-n, \dots, n\} \), is the minimum
        cost of opening exactly \(k'\) facilities in \(t\) and routing exactly
        \(b\) clients down through \(e_t\) (\(b < 0\) would mean that we are
        routing \(-b\) clients up). The cost of routing is counted to the top
        endpoint of \(e_t\).
    \end{definition}
    
\begin{lem}\label{lem:ckm-tree}
The \CKM{} problem on trees admits a polynomial time exact algorithm.
\end{lem}
\begin{proof}
    Computing \(D(t, k', b)\) on \(t\) with two children \(t_1\) and \(t_2\)
    amounts to finding \(k'_1\), \(k'_2\), \(b_1\) and \(b_2\) that minimize
    \[
        D(t_1, k'_1, b_1) + D(t_2, k'_2, b_2),
    \] such that \(b_1 + b_2 = b\) and \(k'_1 + k'_2 = k'\). They can be
    trivially found in \(\OO\br{k \cdot n}\) time for a single pair \(\langle
    k', b\rangle \). Once \(k'_1\), \(k'_2\), \(b_1\) and \(b_2\) are found, we
    set
    \[
        D(t, k', b) = D(t_1, k'_1, b_1) + D(t_2, k'_2, b_2) + d(e_t) \cdot |b|,
    \]
    where \(d(e)\) is the length of the edge in our tree. For a leaf \(l\),
    \(D(l, k', b)\) is defined naturally, depending on whether the leaf holds a
    client or a facility. Note, that for a leaf with a facility, \(D(l, 1, b)\)
    is finite also for \(b\) smaller than the capacity of the
    facility, as the optimal solution might not use it entirely. Finally, the
    optimum solution to the \CKM{} problem on the entire tree \(T\) is equal to
    \[
        \min_{k' \in \{1, \dots, k\}}  D(T, k', 0).
    \]
\end{proof}

\section{Conclusions and open problems}
\mic{We have presented a $\br{7+\eps}$-approximation algorithm for the \CKM problem,
which consists of three building blocks:
approximation for \UKM, metric embedding into a simpler structure, and a parameterized algorithm working on $\ell$-centered instances.

Whereas the first and the last ingredient are almost lossless from the approximation point of view, the embedding procedure seems to be the main bottleneck for obtaining a better approximation guarantee.
One can imagine that a different technique would allow to obtain a $\br{1+\eps}$-approximation in \FPT{} time.
We believe that finding such an algorithm or ruling out its existence is an interesting research direction.

Another avenue for improvement is processing $k$-centered instances in time $2^{\OO(k)}n^{\OO(1)}$.
Such a~routine would reduce the running time of the whole algorithm to single exponential.
In order to do so, one could replace the subroutine for \UKM by Lin and Vitter~\cite{lin92} with a~standard approximation algorithm that opens exactly $k$ facilities, what would moderately increase the constant in approximation ratio.

Finally, whereas we have used the framework of $\ell$-centered instances to devise an FPT approximation,
it might be possible to explore the structure of special instances further and find a~polynomial time approximation algorithm. This could yield an improvement over the $\OO(\log k)$-approximation ratio for \CKM, which remains a major open problem.
}

\marr{some more?}
\marr{maybe we can add that in light of recent developments on approximation in FPT time [Cygan et al.] one could try to show that 1+eps approximation is actually not possible? then again, already apx-hardness is not known so i don't know if this is not even a harder question}

\newpage
\bibliographystyle{abbrv}
\bibliography{kMedianbib}

\end{document}